\documentclass[conference]{IEEEtran}
\IEEEoverridecommandlockouts
\usepackage{cite}
\usepackage{amsmath,amssymb,amsfonts}
\usepackage{algorithmic}
\usepackage{graphicx}
\usepackage{textcomp}
\usepackage{xcolor}
\usepackage{microtype}
\usepackage{bm}
\usepackage{amsthm}
\usepackage{url}
\usepackage{hyperref}

\newtheorem{theorem}{Theorem}
\newtheorem{lemma}{Lemma}

\newcommand{\remove}[1]{}

\newcommand{\Ra}{\Rightarrow}
\newcommand{\ra}{\rightarrow}
\newcommand{\CC}{{L}}
\newcommand{\CB}{{L_B}}
\newcommand{\C}{C}

\usepackage{microtype}

\begin{document}
\title{Removing Sequential Bottleneck of Dijkstra's Algorithm for the Shortest Path Problem\thanks{
Supported by NSF CNS-1812349, CNS-1563544, and the Cullen Trust for Higher Education Endowed Professorship}
}
\author{Vijay K. Garg,\\
  The University of Texas at Austin,\\
  Department of Electrical and Computer Engineering,\\
  Austin, TX 78712, USA}

\bibliographystyle{plainurl}


\maketitle

\begin{abstract}
All traditional methods of computing shortest paths depend upon edge-relaxation where
the cost of reaching a vertex from a source vertex is possibly decreased if that edge is used.
We introduce a method which maintains lower bounds as well as upper bounds for
reaching a vertex. This method enables one to find the optimal cost for multiple
vertices in one iteration and thereby reduces the sequential bottleneck in Dijkstra's algorithm.

We present four algorithms in this paper --- $SP_1$, $SP_2$, $SP_3$ and $SP_4$.
$SP_1$ and $SP_2$ reduce the number of heap operations in Dijkstra's algorithm.
For directed acyclic graphs, or directed unweighted graphs they have the optimal complexity of $O(e)$ where $e$ is the number
of edges in the graph which is better than that of Dijkstra's algorithm.
For general graphs, their worst case complexity matches that of Dijkstra's algorithm
for a sequential implementation but allows for greater parallelism.
Algorithms $SP_3$ and $SP_4$ allow for even more parallelism but with higher work complexity.
Algorithm $SP_3$ requires $O(n + e(\max(\log n, \Delta)))$ work where $n$ is the number of vertices and $\Delta$ is the maximum in-degree
of a node. Algorithm $SP_4$ has the most parallelism. It requires $O(ne)$ work. These algorithms generalize the work 
by Crauser, Mehlhorn, Meyer, and Sanders on parallelizing Dijkstra's algorithm.
%
 \end{abstract}



\section{Introduction}

The single source shortest path (SSSP) problem has wide applications in transportation, networking and many other fields.
The problem takes as input a weighted directed graph with $n$ vertices and $e$ edges.
We are required to find $cost[x]$, the minimum cost of a path from the {\em source} vertex $v_0$ 
to all other vertices $x$ where the cost of a path is defined as the sum of edge weights along that path.
We assume that all edge weights are {\em strictly positive} throughout this paper.

Most SSSP algorithms are inspired by Dijkstra's algorithm \cite{Dijkstra1959} or Bellman-Ford \cite{bellman1958routing,ford56}.
We present four algorithms in this paper in increasing order of work complexity. Algorithms $SP_1$, $SP_2$ and $SP_3$ are inspired by Dijkstra's algorithm
and $SP_4$ is inspired by Bellman-Ford algorithm.
Algorithms $SP_1$ and $SP_2$ are suitable for sequential implementations. They improve upon Dijkstra's algorithm
by reducing the total number of heap operations. 
For acyclic graphs, $SP_1$ performs no heap operations (except for the insertion of the initial source vertex)
and has the time complexity of $O(e)$.
Hence, it unifies Dijkstra's algorithm with the topological sorting based algorithm for acyclic graphs.
$SP_2$ has the optimal time complexity of $O(e)$ whenever the input graph is acyclic or unweighted.
%
For general graphs, their worst case asymptotic complexity matches that of Dijkstra's algorithm for a sequential implementation; however,
they always perform less heap operations than Dijkstra's algorithm.
Additionally, they are more suitable for
a parallel implementation because
they allow multiple vertices to be explored in parallel unlike Dijkstra's algorithm which explores vertices in the order 
of their shortest cost.  Algorithm $SP_2$ allows more parallelism than $SP_1$ at the expense of an additional $O(e)$ processing.

Algorithm $SP_3$ allows for even more parallelism than $SP_2$. It uses the technique of keeping lower bounds
on $cost[x]$ for all vertices $x$.
Almost all algorithms for the shortest path problem are based on keeping upper bounds.
Dijkstra's algorithm keeps $D[x]$, an upper bound on the cost of the path for any vertex $x$.
It maintains the invariant that $D[x]$ always reflects the cost of a feasible path in the directed graph from the source vertex to $x$.
Our algorithm $SP_3$ extends Dijkstra's algorithm by maintaining the variable $\C[x]$ for any vertex $x$ that gives a lower bound on the cost to reach $x$. 
The invariant we maintain is that 
any path from the source vertex to $x$ must have cost at least $\C[x]$. When $\C[x]$ is zero, the invariant is trivially true in a directed graph with
no negative weights. 
At each iteration of the algorithm, we increase $\C[x]$ for one or more vertices till we reach a point where $\C$ is also {\em feasible} and corresponds to 
the cost of all shortest paths. 
The vertices that have matching upper bounds and lower bounds are
called {\em fixed} vertices and the minimum cost from the source vertex to these vertices are known. 
By combining the upper bounds of Dijkstra's algorithm with the lower bounds,
we present an algorithm, $SP_3$, for the single source shortest path algorithm that is superior to Dijkstra's algorithm in two respects.
First, Dijkstra's algorithm suffers from the well-known sequential bottleneck (e.g. \cite{crauser1998parallelization, meyer2003delta}). 
Outgoing edges of only those vertices are explored (relaxed) whose distance
is the minimum of all vertices whose adjacency list have not been explored. In contrast, our algorithm explores all those vertices $x$ whose upper bounds $D[x]$
and lower bounds $\C[x]$ match and have not been explored before. Although the idea of marking multiple vertices fixed in a single iteration
has been explored before (for e.g. \cite{crauser1998parallelization}), this is the first paper, to the best of our knowledge, that marks vertices {\em fixed}
based on the idea of lower bounds.
 Second, when one is interested in a shortest path to a single destination, our algorithm
may determine that $D[x]$ is equal to $\C[x]$ much sooner than Dijkstra's algorithm. 

There are two assumptions in our algorithms. First, we assume that all weights are strictly positive. This is a minor
strengthening of the assumption in Dijkstra's algorithm where all weights are assumed to be nonnegative.
The second assumption is that we have access to incoming edges for any vertex discovered during the execution of the algorithm.
Dijkstra's algorithm uses only an adjacency list of outgoing edges. This assumption is also minor in the context of static graphs. However, 
when the graph is used in a dynamic setting, it may be difficult to find the list of incoming edges. 
We assume in this paper that either the graph is static or that a vertex can be expanded in the backward direction in a dynamic graph.

The single source shortest path problem has a rich history.  One popular research direction is to improve the worst case complexity of Dijkstra's algorithm by
using different data structures. For example, by using Fibonacci heaps for the min-priority queue, Fredman and Tarjan \cite{Fredman:1987:FHU:28869.28874} gave an algorithm that takes $O(e + n \log n)$. There are many algorithms that run faster when weights are small integers bounded by some constant $\gamma$.
For example, Ahuja et al \cite{ahuja1990faster} gave an algorithm that uses Van Emde Boas tree as the priority queue to give an algorithm
that takes $O(e \log \log \gamma)$ time. Thorup \cite{thorup2000ram} gave an implementation that takes $O(n+e \log \log n)$ under special constraints on the weights.
Raman \cite{Raman:1997:RRS:261342.261352} gave an algorithm with $O(e + n  \sqrt{\log n \log \log n})$ time.
Our algorithms do not improve the worst case sequential complexity of the
problem, but reduce the sequential bottleneck. Our algorithms also reduce the number of priority queue operations in the average case.

It is also interesting to compare our approach with algorithm $A^*$ \cite{Hart68}. The algorithm $A^*$ is applicable when there is a single target vertex and
there is a heuristic function $h(x)$ for any vertex that provides the lower bound from $x$ to the target vertex. The heuristic function assumes that
there is some background knowledge that provides the lower bound to the target. 
Our algorithms are not based on a target vertex or availability of the background knowledge.
Even though $A^*$ also uses the notion of lower bounds, the usage is different.
We use the lower bound from the source vertex to $x$ in our algorithms and not the lower bound from $x$ to the target vertex.

There are many related works for parallelizing Dijkstra's algorithm. The most closely related work is 
Crauser et al \cite{crauser1998parallelization} which gives three methods to 
improve parallelism.
These methods,  in-version, out-version and in-out-version, allow multiple vertices to be marked as fixed instead of just the one with the minimum $D$ value.
The in-version marks as fixed any vertex $x$ such that $D[x] \leq \min \{D[y] ~|~ \neg fixed(y) \} + \min \{ w[v,x] ~|~ \neg fixed(x) \}$.
This method is a special case of our algorithm $SP_2$.
The implementation of in-version in \cite{crauser1998parallelization} requires an additional priority queue and the total number of heap operations 
increases by a factor of $2$ compared to Dijkstra's algorithm even though it allows greater parallelism.
 Our algorithm $SP_2$ uses fewer heap operations than Dijkstra's algorithm.
The out-version in \cite{crauser1998parallelization} works as follows.
 Let $L$ be defined as $\min \{ D[x] + w[x,y] ~|~ \neg fixed(x) \}$. Then,
the out-version marks as fixed all vertices that have $D$ value less than or equal to $L$.
Our method is independent of this observation and we incorporate out-version in algorithms
$SP_3$ and $SP_4$.
The in-out-version is just the use of in-version as well as out-version in conjunction.


A popular practical parallel algorithm for SSSP is $\Delta$-stepping algorithm due to Meyer and Sanders \cite{meyer2003delta}.
Meyer and Sanders also provide an excellent review of prior parallel algorithms in \cite{meyer2003delta}. 
They classify SSSP algorithms as either {\em label-setting}, or {\em label-correcting}. Label-setting algorithms, such as Dijkstra's algorithm,
relax edges only for  fixed vertices. Label-correcting algorithms may relax edges even for non-fixed vertices.
$\Delta$-stepping algorithm is a label-correcting algorithm in which eligible non-fixed vertices are kept in an array of buckets
such that each bucket represents a distance range of $\Delta$. During each phase, the algorithm removes all vertices of the first non-empty
bucket and relaxes all the edges of weight at most $\Delta$. Edges of higher weights are relaxed only when their starting vertices are fixed.
The parameter $\Delta$ provides a trade-off between the number of iterations and the work complexity. For example, when $\Delta$ is $\infty$, 
the algorithm reduces to Bellman-Ford algorithm where any vertex that has its $D$ label changed is explored. When $\Delta$ equals $1$ for
integral weights, the algorithm is a variant of Dijkstra's algorithm.
They show that by taking $\Delta = \Theta(1/d)$ where $d$ is the maximum degree of a  graph on $n$ vertices, and random edge weights that are uniformly distributed in 
$[0,1]$, their algorithm takes $O(n+e+dM)$ where $M$ is the maximum shortest path weight from the source vertex to any other vertex.
There are many practical large-scale implementations of the $\Delta$-stepping algorithm (for instance, by
Madduri et al \cite{madduri2007experimental}) in which authors have shown the scalability of the algorithm.
Chakravarthy et al \cite{chakaravarthy2017scalable} give another scalable implementation of an algorithm that is a hybrid of the Bellman-Ford algorithm and
the $\Delta$-stepping algorithm. 
The $\Delta$-stepping technique is orthogonal to our method which is based on keeping lower bounds with vertices. 
It is possible to apply $\Delta$-stepping in conjunction with our method.



In summary, 
we present four algorithms for SSSP in this paper in order of increasing work complexity. 
We only compute the cost of the shortest paths and not the actual paths because
the standard method of keeping backward parent pointers is applicable to all of our algorithms.
Algorithm $SP_1$ counts the number
of incoming edges to a vertex that have been relaxed. When all incoming edges have been relaxed, we show that it is safe to mark this vertex as fixed.
The algorithm $SP_2$ generalizes $SP_1$ to allow even those vertices to be marked as fixed which have incoming edges from non-fixed vertices under certain conditions. Both of these algorithms have fewer heap operations than Dijkstra's algorithm for the sequential case and allow more parallelism when multiple cores are used. 
The algorithm $SP_3$ generalizes $SP_2$ further by maintaining the lower bound $C$ for each vertex.
All these algorithms are label-setting. Algorithm $SP_3$ has the same asymptotic complexity as Dijkstra's algorithm when the 
maximum in-degree of a vertex is $O(\log n)$.
 It allows even more parallelism than $SP_2$.
The algorithm $SP_4$ is a label-correcting algorithm. It has the the most parallelism but with highest work complexity.
$SP_4$ combines ideas from Bellman-Ford, Dijkstra, \cite{crauser1998parallelization} and $SP_3$
for faster convergence of $D$ and $C$ values. 

\section{Background and Notation}


Dijkstra's algorithm (or one of its variants) is the most popular single source shortest path algorithm used in practice.
For concreteness sake we use the version shown in Fig. \ref{fig:dijk} for comparison with our algorithm.
The algorithm also helps in establishing the terminology and the notation used in our algorithm.
We consider a directed weighted graph $(V, E, w)$ where $V$ is the set of vertices, $E$ is the set of directed edges and
$w$ is a map from the set of edges to positive reals (see Fig. \ref{fig:myGraph3} for a running example).
To avoid trivialities, we assume that the graph is loop-free and every vertex  $x$, except the source vertex $v_0$,  has 
at least one incoming edge.


{\small
\begin{figure}[htb]\begin{center}
\fbox{\begin{minipage}  {\textwidth}\sf
\begin{tabbing}
\=x\=xxxx\=xxxx\=xxxx\=xxxx\= \kill
\>\> {\bf var} $D$: array[$0 \ldots n-1$] of integer\\
\>\>\>\>  initially $\forall i: D[i] = \infty$;\\
\>\> \> $fixed$: array[$0 \ldots n-1$] of boolean \\
\> \> \> \> initially $\forall i: fixed[i] = false$;\\
\> \> \> $H$: binary heaps of $(j,d)$ initially empty;\\
\> \> $D[0] := 0$;\\
\> \> $H$.insert((0,D[0]));\\
 \>      \> {\bf while} $\neg H$.empty() {\bf do}\\
  \> \> \> $(j,d) := H$.removeMin();\\
 \> \> \> $fixed[j] := true$;\\
 \> \> \> {\bf forall} $k$: $\neg fixed(k) \wedge (j,k) \in E$ \\
  \> \> \> \> if ($D[k] > D[z]+w[z,k]$) then \\
\> \> \> \> \>$D[k] := D[z] + w[z,k]$;\\
 \> \> \> \> \> $H$.insertOrAdjust $(k, D[k])$;\\
     \>  \> {\bf endwhile};
\end{tabbing}
\end{minipage}
} 
\end{center}
\caption{Dijkstra's algorithm to find the shortest costs from $v_0$ \label{fig:dijk}.}
\end{figure}
}

\begin{figure}[htb]
\begin{center}
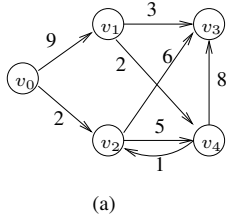
\caption{\label{fig:myGraph3}A Weighted Directed Graph}
\end{center}
\end{figure}

Dijkstra's algorithm maintains $D[i]$, which is a tentative cost to reach $v_i$ from $v_0$. 
Every vertex $x$  in the graph has initially $D[x]$ equal to $\infty$. 
Whenever a vertex is discovered for the first time, its $D[x]$ becomes less
than $\infty$. We use the predicate
$discovered(x) \equiv D[x] < \infty$.
The variable $D$ decreases for a vertex whenever a shorter path is found due to edge relaxation.


In addition to the variable $D$, a boolean array {\em fixed} is maintained. Thus, every discovered vertex is either
{\em fixed} or {\em non-fixed}.
The invariant maintained by the algorithm is that
if a vertex $x$ is {\em fixed} then $D[x]$ gives the final shortest cost from vertex $v_0$ to $x$.
If $x$ is {\em non-fixed}, then $D[x]$ is the cost of the shortest path to $x$ that goes only through fixed vertices.
%

A heap $H$ keeps all vertices  that have been discovered but
are non-fixed along with their distance estimates $D$.
We view the heap as consisting of tuples of the form $(j, D[j])$ where the heap property is with respect to $D$ values.
The algorithm has one main {\em while} loop that removes the vertex
with the minimum distance from the heap with the method $H$.removeMin(), say $v_j$, and marks it as fixed.  It then {\em explores} the vertex $v_j$ by relaxing
all its adjacent edges going to non-fixed vertices $v_k$.
The value of $D[k]$ is updated to the minimum of
$D[k]$ and $D[j]+w[j,k]$. 
If $v_k$ is not in the heap, then it is inserted, else
if $D[k]$ has decreased then the label associated with vertex $k$ is adjusted in the heap. We abstract this step
as the method $H$.insertOrAdjust$(k, D[k])$. 
The algorithm terminates when the heap is empty.
At this point there are no discovered non-fixed vertices and $D$ reflects the cost of the shortest path to all discovered vertices. 
If a vertex $j$ is not discovered then $D[j]$ is infinity
reflecting that $v_j$ is unreachable from $v_0$.

Observe that every vertex goes through the following states. Every vertex $x$ is initially {\em undiscovered} (i.e., $D[x] = \infty$).
If $x$ is reachable from the source vertex, then it is eventually {\em discovered} (i.e., $D[x] < \infty)$. 
A discovered vertex is initially {\em non-fixed}, and is therefore in the heap $H$.
Whenever a vertex is removed from the heap it is a {\em fixed} vertex. A fixed vertex may either be {\em unexplored} or {\em explored}.
Initially, a fixed vertex is unexplored. It is considered explored when all its outgoing edges have been relaxed.

The following lemma simply summarizes the well-known properties of Dijkstra's algorithm.
\begin{lemma}\label{lem:dijk}
The outer loop in Dijkstra's algorithm satisfies the following invariants.\\
(a) For all vertices $x$: $fixed[x] \Rightarrow (D[x] = cost[x])$.\\
(b) For all vertices $x$:
$D[x]$ is equal to cost of the shortest path from $v_0$ to $x$ such that all vertices in the path before $x$ are fixed.\\
(c) For all vertices $x$:
$x \in H$ iff $discovered(x) \wedge \neg fixed[x]$.
\end{lemma}

\section{Algorithm $SP_1$: Using Predecessors}

Dijkstra's algorithm finds the vertex with the minimum tentative distance and
marks it as a fixed vertex. This is the only mechanism by which a vertex
is marked as {\em fixed} in Dijkstra's algorithm. Finding the non-fixed vertex with the minimum $D$ value
takes $O(\log n)$ time when a heap or its variant is used.
Our first observation is that if for any non-fixed vertex $x$, if all the incoming edges are from fixed vertices, then
the current estimate $D[x]$ is the shortest cost. 
To exploit this observation, 
we maintain with each vertex $i$, a variable $pred[i]$ that keeps the number of incoming edges
that have not been relaxed. The variable $pred[i]$ is decremented whenever an incoming edge to vertex $i$  is relaxed.
When $pred[i]$ becomes zero, vertex $i$ becomes fixed. Determining a vertex to be fixed by this additional method increases the rate
of marking vertices as fixed in any iteration of the while loop.

The second observation is that in Dijkstra's algorithm vertices are explored only in order of their cost.
$SP_1$ explores vertices whenever it finds one that is fixed.
Hence, in addition to 
the heap $H$, we maintain a set $R$ of vertices which have been fixed but not explored, i.e., their adjacency lists have not been traversed.
We also relax the invariant on the heap $H$. In Dijkstra's algorithm, the heap does not contain fixed vertices. In algorithm $SP_1$,
the heap $H$ may contain both fixed and non-fixed vertices. However, only those fixed vertices  which
have been {\em explored} may exist in the heap.

{\small
\begin{figure}[htb]\begin{center}
\fbox{\begin{minipage}  {\textwidth}\sf
\begin{tabbing}
\=x\=xxxx\=xxxx\=xxxx\=xxxx\=xxxx\=xxxx\=\kill
\>\> {\bf var} $D$: array[$0 \ldots n-1$] of integer \\
\> \> \> \> initially $\forall i: D[i] = \infty$;\\
\> \> \> $H$: binary heap of $(j,d)$ initially empty;\\
\>\> \> $fixed$: array[$0 \dots n-1$] of boolean\\
\> \>\>\> initially $\forall i: fixed[i] = false$;\\
\> \> \> $Q, R$: set of vertices initially empty;\\
\> \> \> $pred$: array[$0 \ldots n-1$] of integer\\
\> \>\>\> initially $\forall i: pred[i] = ~|~ \{ x ~|~ (x, v_i) \in E \} ~|~$;\\

\\
\> \> $D[0] := 0$;\\
\> \> $H$.insert$((0, D[0]))$;\\
 \>      \> {\bf while} $\neg H$.empty() {\bf do}\\
\> \> \>  $(j,d) := H$.removeMin();\\
\> \> \>  if ($\neg fixed[j]$) then  \\
\> \> \> \> $R$.insert($j$); \\
\> \>\> \> $fixed[j]$ := $true$; \\
\>  \> \>  \> {\bf while} $R \neq \{\}$ {\bf do}\\
\>  \> \>  \> \> {\bf forall} $z \in R$\\
\> \>\>   \> \> \>  $R$.remove($z$);\\
\> \>  \> \>  \> \> {\bf forall} $k: \neg fixed(k) \wedge (z,k) \in E$:\\
 \> \>  \> \>  \> \> \> processEdge1($z,k$);\\
\> \>    \>  \> {\bf endwhile};\\
\> \> \>  \>  {\bf forall} $z \in Q$:\\
 \> \> \>  \> \> $Q$.remove($z$);\\
 \> \> \>  \> \>  if $\neg fixed[z]$ then\\
 \> \> \>  \> \> \>  $H$.insertOrAdjust $(z, D[z])$;\\
 
  \>  \> {\bf endwhile};\\
  \\
\> {\bf procedure} processEdge1($z,k$);\\
\>\> {\bf var} $changed$: boolean initially false;\\
\> \> $pred[k] := pred[k]-1$;\\
 \> \> if ($D[k] > D[z]+w[z,k]$) then \\
 \> \> \> $D[k] := D[z] + w[z,k]$;\\
 \> \> \>   $changed$ := true; \\


 \> \> if $(pred[k] = 0)$ then\\
 \> \> \>    $fixed[k] := true$;\\
 \> \> \>  $R$.insert($k$); \\
   \>\>  else if   ($changed \wedge (k \not \in Q)$) then\\
   \> \> \>  $Q$.insert($k$);
\end{tabbing}
\end{minipage}
} 
\end{center}
\caption{Algorithm $SP_1$  \label{fig:alg-sp1}}
\end{figure}
}

The algorithm $SP_1$ is shown in Fig. \ref{fig:alg-sp1}. The algorithm starts with the insertion of the source vertex with its $D$ value as $0$ 
in the heap.
Instead of removing the minimum vertex from the heap in each iteration and then exploring it, the algorithm consists
of two {\em while} loops. The outer while loop removes one vertex from the heap.
If this vertex is fixed, then it has already been explored and therefore it is skipped; otherwise,
it is marked as fixed and inserted in $R$ to start the inner while loop. The inner loop keeps processing the set $R$ till
it becomes empty. 

We do not require that
vertices in $R$ be explored in the order of  their cost. If $R$ consists of multiple vertices then all of them can be explored
in parallel. During this exploration other non-fixed vertices may become fixed. These are then added to $R$.
Some vertices may initially be non-fixed but eventually while processing $R$ may become fixed.
To avoid the expense of inserting these vertices in the heap, we collect all such vertices which may need to be inserted or adjusted
in the heap in a separate set called $Q$. Only, when we are done processing $R$, we call $H.insertOrAdjust$ on vertices in $Q$.

The vertices $z \in R$ are explored as follows. We process all out-going adjacent edges $(z,k)$ of the vertex $z$ to non-fixed vertices $k$.
This step is called {\em processEdge1} in Fig. \ref{fig:alg-sp1}.
First, we decrement the count $pred[k]$ to account for its predecessor $z$ being fixed.
Then, we do the standard edge-relaxation procedure by checking whether $D[k]$ can be decreased by taking this edge.
If $pred[k]$ is zero, 
$k$ is marked as fixed. Setting $fixed[k]$ to true also removes it effectively from the heap
because whenever a fixed vertex is extracted in the outer while loop it is skipped.

Finally, if $D[k]$ has decreased and $pred[k]$ is greater than $0$, we insert it in $Q$ so that 
once $R$ becomes empty we can call $H$.insertOrAdjust() method on vertices in $Q$



Consider the graph in Fig. \ref{fig:myGraph3}(a). Initially $(0,D[0])$ is in the heap $H$. 
Since there is only one vertex in the heap $H$, it is also the minimum. This vertex is removed and inserted in $R$ marking
$v_0$ as fixed. Now, outgoing edges of $v_0$ are relaxed. Since $pred[1]$ becomes $0$, $v_1$ is marked as fixed and added to $R$.
The vertex $v_2$ has $pred$ as $1$ and $D[2]$ as $2$ after the relaxation of edge $(v_0, v_2)$. The vertex $v_2$ is inserted in the $Q$ for later insertion in the heap. Since $R$ is not empty, 
outgoing edges of $v_1$ are relaxed.  The vertex $v_3$ is inserted in $Q$ and its $D$ value is set to  $12$. The vertex $v_4$ is  also inserted in $Q$
and its $D$ value is set to $11$. At this point $R$ is empty and we insert vertices in $Q$ in $H$ and get back to the outer while loop. 
The minimum vertex $v_2$ is removed from the heap, marked as fixed, and
inserted in $R$ for exploration.
When $v_2$ is explored, the $D$ label of $v_3$ is adjusted to $8$. When edge $(v_2, v_4)$ is relaxed, $D[4]$ is reduced to $7$.
Moreover, $pred[4]$ becomes zero and $v_4$ is inserted in $R$ for exploration. When $v_4$ is explored, $pred[3]$ also becomes zero
and is also inserted in $R$. Once $v_3$ is explored, $R$ becomes empty. We then go to the outer while loop. All vertices in the heap are
fixed and therefore the algorithm terminates with $D$ array as $[0,9,2,8,7]$. Observe that it is easy to maintain a count of the non-fixed vertices in the
the heap and the method $H.empty()$ can be overloaded to return true whenever this count is zero.

We now show that
\begin{lemma}\label{lem:pred}
Let $v$ be any non-fixed vertex. Suppose all incoming edges of $v$ have been relaxed, then $D[v]$ equals $cost[v]$.
\end{lemma}
\begin{proof}
We show that whenever $pred[v]$ is zero, 
$D[v]$ equals $cost[v]$.  We prove this lemma by contradiction.
 If not, let $x$ be the vertex with the smallest $D$ value such that
all its incoming edges have been relaxed but $D[x]$ is greater than $cost[x]$. Let 
 $\alpha$ be a path from $v_0$ to $x$ with the smallest cost (and therefore less than $D[x]$). 
The path $\alpha$ must go through
a non-fixed vertex because $D[x]$ is the minimum cost of all paths that go through fixed vertices. Let $y$ be the last non-fixed vertex along this path.
The successor of $y$ in that path cannot be  $x$ because all predecessors of $x$ are fixed.
Therefore, its successor is a fixed vertex $z$ because $y$ is the last non-fixed vertex along the path.
The path $\alpha$ can be broken into two parts --- the path from the source vertex to $z$ and then the path from $z$ to $x$.
The path from $z$ to $x$ consists only of fixed vertices by the definition of $y$. It is sufficient to show that there exists a path from the source vertex to 
$z$ that consists only of fixed vertex with the same cost as in $\alpha$.
The vertex $z$ can be fixed either because it has the minimum value of $D$ in the heap at some iteration, or
because all the incoming edges to $z$ have been relaxed. In the former case, $D[z] = cost[z]$ and therefore there exists a 
path from the source vertex to $z$ with only fixed vertices and the minimum cost. In the latter case, when $z$ is fixed because all its incoming edges have
been relaxed, then by our choice of $x$, $D[z]$ is equal to $cost[z]$ which again shows
existence of a path with only fixed vertices with the minimum cost.
%
\end{proof}

To show the correctness of Algorithm $SP_1$, we make the following claims.
We use the predicate $explored(x)$ that holds true iff the adjacency list of $x$ has been explored.
\begin{lemma}\label{lem:SP1-outer}
The following invariants hold at the outer and the inner while loop of $SP_1$.\\
(a) For all vertices $x$: $fixed[x] \Rightarrow D[x] = cost[x]$.\\
(b) For all vertices $x$:
$D[x] = $ cost of the shortest path from $v_0$ to $x$ such that all vertices in the path before $x$ are fixed.\\
(c) For all vertices $x$:
$x \in H \Ra discovered(x) \wedge (\neg fixed[x] \vee explored(x)) $.\\
Furthermore, $\forall x: discovered(x) \wedge \neg fixed[x] \Ra (x \in H)$.\\
\end{lemma}
\begin{proof}
(a,b) The only difference from Dijkstra's algorithm is that in one iteration of the outer while loop, not only vertices with the minimum value
of $D$ are fixed, but also vertices with $pred[x]$ equal to $0$. Due to Lemma \ref{lem:pred}, the invariant on $fixed$ and $D$ continues to hold.
In the inner loop, whenever a vertex is discovered and is not fixed, it is inserted in the heap maintaining the invariant on $H$.\\
(c) Whenever a vertex $x$ is discovered and it is not fixed, it is inserted in the heap. Whenever a vertex is removed from the heap
it is marked as fixed. A vertex in the heap can also become fixed in the inner while loop. However, whenever a vertex becomes fixed it
is inserted in $R$ for exploration and $R$ is empty at the outer while loop. Hence, any vertex that is fixed is also explored.
\end{proof}

\begin{lemma}\label{lem:SP1-inner}
The following invariant holds at the inner while loop of $SP_1$.
For all vertices $x$:
$x \in R$ iff $fixed[x] \wedge \neg explored(x)$.
\end{lemma}
\begin{proof}
Whenever a vertex is marked fixed initially, it is inserted in $R$. Whenever it is explored, it is removed from $R$.
\end{proof}

We now have the following Theorem.
\begin{theorem}
Algorithm $SP_1$ returns the weight of a shortest path from source vertices to all other vertices.
\end{theorem}
\begin{proof}
Consider any vertex $x$ reachable from the source vertex. We show that $x$ is eventually discovered.
We use  induction on $k$ equal to the number of vertices with cost that  less than or equal to that of $x$.
The base case is trivial. For the inductive case, $x$ has at least one predecessor.
Since all weights are positive, all predecessors of $x$ have cost less than that of $x$.
If all vertices are sorted based on their cost, the outer while loop marks as fixed at least one vertex with cost less than or equal to $x$.
Hence, in at most $k$ iterations of the outer while loop, one of the predecessors of $x$ is marked as fixed. The algorithm terminates only when 
every fixed vertex is explored, and therefore $x$ is discovered.

Any vertex $x$ that is discovered is either in $H$ when it is not fixed or fixed but explored, or in $R$ when it is fixed and not explored.
If the vertex $x$ is $fixed$, from the invariant on $D$ and $fixed$, we have that $D[x]$ equal to $cost[x]$.
If the vertex $x$ is not $fixed$, it is eventually removed from the heap $H$ and becomes fixed.
Hence, any reachable vertex $x$ has its $D[x]$ set to $cost[x]$.

If any vertex $x$ is not reachable, then it can never be discovered and $D[x]$ returns $\infty$ due to initialization.
\end{proof}

We now show that $SP_1$ cuts down the complexity of Dijkstra's algorithm significantly for acyclic graphs whenever source vertex
is the only vertex with no incoming edges.
To ensure this, whenever we read the graph 
we create a list $L$ of all vertices other than the source vertex that have no incoming edges. All these vertices are clearly not reachable from the source vertex.
We then repeatedly remove vertices from the list $L$ and their outgoing edges from the graph. If in this process, another vertex has all its
incoming edges removed, it is added to the list $L$. The procedure is continued until $L$ becomes empty and we are guaranteed
that the source vertex is the only vertex with no incoming edges. This procedure takes at most $O(e)$ time because
any edge is processed at most once. 

We now have the following result.
\begin{theorem}
$SP_1$ takes $O(e + n \log n)$ time with Fibonacci heaps for any directed graph and takes $O(e)$ time for directed acyclic graphs 
in which source node is the only one with with zero incoming edges.
\end{theorem}
\begin{proof}
For a general directed graph, any steps taken in $SP_1$ is also
taken in Dijkstra's algorithm except for the constant time operations such as decrementing $pred$, and inserting or deleting a vertex from $Q$ and $R$.
Both $Q$ and $R$ can be implemented as a linked lists with $O(1)$ insertions at the tail and $O(1)$ deletions at the head of the list.
The membership in $Q$ can also be implemented in $O(1)$ time using a bit vector.
Hence, using Fibonacci heaps, we get the time complexity of Dijkstra's algorithm.

For directed acyclic graphs, initially the source vertex is removed from the heap and inserted in $R$. Now as we explore $R$, the predecessor count for
all vertices adjacent to the source vertex will decrease by $1$. Since the graph is acyclic, at least one new vertex will become fixed.
As we continue processing $R$, all the nodes of the graph will become fixed (just as in the topological sort of an acyclic graph).
Thus, all
reachable vertices of an acyclic graph will be processed in the first iteration of the outer while loop. In this iteration, every edges is processed exactly
once, giving us $O(e)$ time complexity.
\end{proof}
The worst case for $SP_1$ is when the vertex discovered last has outgoing edges to all other vertices. In such a worst-case scenario,
$SP_1$ will not have any vertex becomes fixed through processing of $R$ and the algorithm will degenerate into Dijkstra's algorithm.

\section{Algorithm $SP_2$: Using Weights of Incoming edges}

We now strengthen our mechanism to mark vertices as fixed. $SP_2$ requires access to incoming edges for any vertex.
Let a vertex $k$ be discovered from a predecessor vertex $z$.
Then, we compute $inWeight[k]$ as the minimum weight of incoming edges from all predecessors  other than $z$.
We exploit $inWeight$ as follows.
\begin{lemma}\label{lem:inWeight}
Let $k$ be any non-fixed vertex discovered from the vertex $z$ in any iteration of the outer while loop with $d$.
If ($D[k] \leq d + inWeight[k]$) then $D[k]$ equals $cost[k]$.
\end{lemma}
\begin{proof}
Since $d$ is the weight of the vertex removed from the heap $H$, we know that any predecessor vertex $v$ that is not fixed is guaranteed to
have $D[v] \geq d$. Hence, $D[k]$ is guaranteed to be less than or equal to $D[v] + w[v,k]$ for any incoming edge $(v,k)$ that is relaxed.
\end{proof}

This mechanism comes at the space overhead of maintaining an additional array $inWeight[]$ indexed by vertices.

{\small
\begin{figure}[htb]\begin{center}
\fbox{\begin{minipage}  {\textwidth}\sf
\begin{tabbing}
\=xx\=xxx\=xxxx\=xxxx\=xxxx\=xxxx\=xxxx\=\kill
\>\> $inWeight$: array [$0 \ldots n-1$] of int\\
\>\>\>  initially $\forall i: inWeight[i] = \infty$;\\

\>\> {\bf procedure} processEdge2($z,k$);\\
\>\> {\bf var} $changed$: boolean initially false;\\
 \> \> $pred[k] := pred[k] - 1$;\\
 \\
 \> \> // Step 1: vertex $k$ has been discovered.\\
\>\>  //  Compute $inWeight$\\
 \> \> if $(D[k] = \infty) \wedge (pred[k] >0)$ then\\
 \> \> \> $inWeight[k] := \min \{ w[v,k] ~|~ (v,k) \in E, v \neq z \}$;\\
 \\
 \> \> // Step 2: relax $(z, k)$ edge\\
 \> \> if ($D[k] > D[z]+w[z,k]$) then \\
 \> \> \>$D[k] := D[z] + w[z,k]$;\\
 \> \> \>$changed$ := true;\\
 \\
  \> \> // Step 3: check if vertex $k$ can be fixed.\\
 \> \> if $((pred[k] = 0) \vee (D[k] \leq d + inWeight[k])$ then\\
 \> \> \> $fixed[k] := true$;\\
 \> \> \>  $R := R$.insert$(k)$; \\
  \>\>  else if   ($changed \wedge (k \not \in Q)$) then $Q$.insert$(k)$;
\end{tabbing}
\end{minipage}
} 
\end{center}
\caption{Algorithm $SP_2$: Algorithm $SP_1$ with {\em processEdge2}  \label{fig:alg-sp2}}
\end{figure}
}

After incorporating Lemma \ref{lem:inWeight}, we get the algorithm $SP_2$ shown in Fig. \ref{fig:alg-sp2}. It is same as $SP_1$ except
we use the procedure $processEdge2$ instead of $processEdge1$.
In step 1, we compute $inWeight[k]$ when it is discovered for the first time, i.e., when $D[k]$ is $\infty$.
If there are additional incoming edges, i.e., $(pred[k] > 0)$, we determine the minimum of all the incoming weights except from the vertex $z$ that 
discovered $k$.
In step 2, we perform the standard edge-relaxation.
In step 3, we check if the vertex $k$ can be fixed either because it has no more predecessors, or 
for any non-fixed predecessor $v$,  the relaxation of the edge $(v,k)$ will not change $D[k]$.
Observe that for sequential implementations, if $R$ is maintained as a queue and all edge weights are uniform, then any vertex
discovered for the first time will always be marked as fixed and will never be inserted in the heap.
For such inputs, $SP_2$ will behave as a simple breadth-first-search.

Since any vertex is discovered at most once, computing $inWeight$ requires processing of all incoming
edges of a vertex at most once. Hence, the cumulative time overhead is linear in the number of edges.
If the graph is unweighted, then $SP_2$ is much faster than Dijkstra's algorithm
when $R$ is implemented as a queue.

\begin{theorem}
Suppose that $R$ is implemented as a simple queue.
$SP_2$ takes
\begin{itemize}
\item  $O(e + n \log n)$ time with Fibonacci heaps for any directed graph, 
\item $O(e)$ time for directed acyclic graphs 
in which only the source node has zero incoming edges,
\item $O(e)$ time for any unweighted directed graph.
\end{itemize}
\end{theorem}
\begin{proof}
Since $SP_2$ retains all properties of $SP_1$, we only need to prove the claim on unweighted directed graphs.
In unweighted directed graphs, once the source vertex is explored any vertex $k$  adjacent to the source vertex
become fixed because it satisfies the condition that $D[k] \leq d + inWeight[k]$ and is inserted in $R$.
Continuing in this manner, the algorithm reduces to breadth-first search by simply inserting nodes in $R$ in breadth-first manner
and removing from $R$ till all reachable vertices
are explored. 
\end{proof}

Hence, $SP_2$ unifies Dijkstra's algorithm with the topological sort for acyclic graphs as well as the breadth-first search for unweighted graphs.
Consequently, it is faster than Dijkstra's algorithm when the input graph is close to an acyclic graph (i.e., has few cycles)
or close to an unweighted graph (most weights are the same).

Lemma \ref{lem:inWeight} is similar to the in-version method of \cite{crauser1998parallelization}.
The in-version fixes any vertex $k$ such that $D[k] \leq d + \min \{ w[j,k] ~|~ \neg fixed(j), (j, k) \in E \}$.
There are two differences. First, we do not include the weight of the edge that discovered $k$ in our calculation 
of $inWeight$. Second, in \cite{crauser1998parallelization} the implementation is based on maintaining an additional priority queue which adds the overhead of
$O(e \log n)$ to the algorithm with ordinary heap implementation. $SP_2$ adds a cumulative overhead of $O(e)$.
In sequential implementations, the in-version increases the number of heap operations, whereas $SP_2$ decreases this number.

Consider the graph in Fig. \ref{fig:myGraph3}(a). Initially $(0,0)$ is in the heap $H$. It is removed and inserted in $R$ marking
$v_0$ as fixed. Now outgoing edges of $v_0$ are relaxed. Since $pred[1]$ becomes $0$, $v_1$ is marked as fixed and added to $R$.
The vertex $v_2$ has $pred[2]$ as $1$ after the relaxation. It is inserted in the heap with $D$ value as $2$ and $inWeight[2]$ is computed as $1$.
Since $R$ is not empty, 
outgoing edges of $v_1$ are relaxed.  The vertex $v_3$ is inserted in $Q$ with $D$ value $12$ and the vertex $v_4$ is inserted
with $D$ value as $11$. We also compute $inWeight[3]$ as $\min\{6, 8\}$ equal to $6$ and $inWeight[4]$ as $5$.
At this point $R$ is empty and the minimum vertex $v_2$ is removed from the heap and marked as fixed.
When $v_2$ is explored and the edge $(v_2, v3)$ is relaxed, the label of $v_3$ is adjusted to $8$. Since $8$ is less than or equal to $d+inWeight[3]=2+6$, it is marked as fixed and inserted in $R$.
When edge $(v_2, v_4)$ is relaxed, $D[4]$ is reduced to $7$.
Moreover, $pred[4]$ becomes zero and $v_4$ is inserted in $R$ for exploration. 
At this point, all vertices are fixed. When $R$ is processed, there are no additional changes and the algorithm terminates
with the $D$ array as $[0,9,2,8,7]$.

\section{Algorithm $SP_3$: Using Lower Bounds with Upper Bounds}
We now generalize the mechanism of $SP_2$ further to determine fixed vertices based on the idea of using lower bounds.
%
The idea of starting with the infinite cost as an estimate of the actual cost
and decreasing the estimate via edge-relaxation has been the underlying principle for not only Dijkstra's algorithm but almost all other shortest path algorithms
such as Bellman-Ford, Floyd-Warshall \cite{floyd1962algorithm} and their derivatives. In this section, we present the idea of using
lower bounds  $C$ associated with every vertex in addition to the upper bounds given by $D$.

We keep a global array $C$ such that $C[x]$ is the lower bound associated with
each vertex $x$. 
We maintain the invariant that there is no path of cost strictly lower than $C[x]$
from the source vertex to $x$. Just as $D[i]$ is initialized to $\infty$, $C[i]$ is initialized to $0$ for all $i$ so that the invariant is true
initially. Clearly, any vertex $x$ such that $C[x]$ and $D[x]$ are equal has both of them equal to $cost[x]$.
Hence, any vertex with $C[x]$ equal to $D[x]$ can be marked as fixed.
Conversely, if any vertex $x$ is known to be fixed (for example, by removal from the min-heap), then $C[x]$ can be set to $D[x]$.

How do we determine nontrivial $C[x]$ for non-fixed vertices? Just as the exploration of a vertex $x$ in Dijkstra's algorithm updates
$D[y]$ for all out-going edges $(x,y)$, we define a dual step that can update $C[x]$ based upon all in-coming edges.
The value of $C[x]$ for the source vertices is always zero. For other vertices, we have
\begin{lemma}\label{lem:crule}
Let $C[x]$ be a lower bound on the cost of the shortest path to $x$. Then, for any vertex $x$ that is not a source vertex,
\begin{equation}\label{eqn:crule}
C[x] \geq \min \{ C[v] + w[v,x] ~|~ (v,x) \in E\} 
\end{equation}
\end{lemma}
\begin{proof}
Since $x$ is not the source vertex, it must have a predecessor $v$ in a shortest path from the source vertex to $x$.
The equation follows by noting that an additional cost of $w[v,x]$ would be incurred as the last edge on that path.
\end{proof}
The lemma gives an alternate short proof of Lemma \ref{lem:pred}. 
Consider any $x$ such that all its predecessors are fixed.
Since $C[v]$ is equal to $D[v]$ for all fixed vertices, from Eqn \ref{eqn:crule},
we get that $C[x] \geq \min \{ D[v] + w[v,x] ~|~ (v,x) \in E \}. $ We also get that
$D[x] \leq \min \{ D[v] + w[v,x] ~|~ (v,x) \in E \}$ using the edge-relaxation rule.
Combining these two inequalities with $C[x] \leq D[x]$, we get that $C[x]$ is equal to $D[x]$ and therefore
$x$ can be marked as fixed.

An additional lower bound on the cost of a vertex is determined using the global information on the graph.
At any point in execution of the graph, there are two sets of vertices --- fixed and non-fixed. Any path from the source vertex
to a non-fixed vertex must include at least one edge from the set of edges that go from a fixed vertex to a non-fixed vertex.
\begin{lemma}\label{lem:crule2}
For any $x$ such that $\neg fixed[x]$,\\
$C[x] \geq \min \{ C[u] + w[u,v] ~|~ (u,v) \in E \wedge fixed[u] \wedge \neg fixed[v] \}$.
\end{lemma}
\begin{proof}
Consider the shortest path from $v_0$ to $x$. Since $fixed[v_0]$ and $\neg fixed[x]$ there is an edge in the path from 
a  fixed vertex $u'$ to a non-fixed vertex $v'$. We get that $C[x] \geq C[v']$ and $C[v']  \geq \min  \{ C[u] + w[u,v] ~|~ (u,v) \in E \wedge fixed[u] \wedge \neg fixed[v] \}$.
\end{proof}

Since for a fixed vertex $u$, $C[u]$ equals $D[u]$, we get that for any non-fixed vertex $x$, 
$C[x] \geq \min \{ D[u] + w[u,v] ~|~ (u,v) \in E \wedge fixed[u] \wedge \neg fixed[v] \}$.
The right hand side is simply the minimum key in the min-heap $H$.

%
%
Finally, we also exploit the
method of \cite{crauser1998parallelization}.
\begin{lemma}\cite{crauser1998parallelization}\label{lem:crauser}
Let $threshold = \min \{ D[u] + w[u,v] ~|~ (u,v) \in E \wedge  \neg fixed[u] \}$.
Consider any non-fixed vertex $x$ with $D[x] \leq threshold$.
Then, $x$ can be marked as a fixed vertex.
\end{lemma}
\begin{proof}
Since $D[x] \leq threshold$, we know that $x$ is a discovered vertex and there is a path from $v_0$ to $x$. We show that this path has the shortest cost.
Suppose that 
there is another path with cost less than $D[x]$. This path must go through at least some non-fixed vertex because $D[x]$ already accounts
for all paths that go through only fixed vertices. Let $u'$ be the first non-fixed vertex on that path.
Then, the cost of that path is at least $threshold$ by the definition of $threshold$ giving us the contradiction.
\end{proof}

This lemma also allows us to mark multiple vertices as fixed and therefore  update $C$ for them.

To exploit Lemma \ref{lem:crauser}, we use two additional variables in $SP_3$.
The variable $outWeight[x]$ keeps the weight of the minimum outgoing edge from $x$.
This array is computed exactly once with the cumulative overhead of $O(e)$.
We also keep an additional binary heap $G$ as proposed in \cite{crauser1998parallelization}.
This heap keeps $D[u]+outWeight[u]$ for all non-fixed vertices. Clearly, the minimum value
of this heap is the required threshold. 


{\small
\begin{figure}[htbp]\begin{center}
\fbox{\begin{minipage}  {\textwidth}\sf
\begin{tabbing}
\=x\=xxx\=xxx\=xxx\=xxx\=xxx\=xxx\=\kill
\>\> {\bf var} $C,D$: array[$0 \ldots n-1$] of integer\\
\>\>\>  initially $\forall i: (C[i] = 0) \wedge (D[i] = \infty)$;\\
\> \> \> $G,H$: binary heap of $(j,d)$ initially empty; \\
\>\> \> $fixed$: array[$0 \ldots n-1$] of boolean \\
\>\>\>\>initially $\forall i: fixed[i] = false$;\\
\> \> \> $Q,R$: set of vertices initially empty;\\
\> \> \> $outWeight$: array[$0 \ldots n-1$] of integer initially\\
\>\>\>\> $\forall i: outWeight[i] = \min \{ w[i,j] ~|~ (i,j) \in E\}$;\\
\> \> $D[0] := 0$;\\
\> \> $H$.insert$(0, D[0])$; $G$.insert$(0, D[0] + outWeight[0])$;\\
\>      \> {\bf while} $\neg H$.empty() {\bf do}\\
 \> \> \> int threshold := $G$.getMin(); \\
 \>   \>   \> {\bf while} ($H$.getMin() $\leq threshold)$ {\bf do} \\
\> \> \> \>  $(j,d) := H$.removeMin();\\
\> \> \> \>  $G$.remove$(j)$; \\   
\>\>\> \> $fixed[j]$ := $true$; \\
\>\> \> \> $C[j] := D[j]$;\\
\> \> \> \>  $R$.insert($j$); \\
\> \> \> \>  if ($H$.empty()) break; \\
 \>    \>  \> {\bf endwhile};\\
  \> \>  \> {\bf while} $R \neq \{\}$ {\bf do}\\
  \> \>  \> \> {\bf forall} $z \in R$\\
\>\>   \> \> \>  $R := R - \{ z \}$\\
\>  \> \>  \> \> {\bf forall} $k$: $\neg fixed(k) \wedge (z,k) \in E$\\
\>  \> \>  \> \> \> processEdge3(z,k);\\
 \>    \>  \> {\bf endwhile};\\
  \> \>  \>  {\bf forall} $z \in Q$:\\
  \> \>  \> \> $Q$.remove($z$);\\
  \> \>  \> \>   if $\neg fixed[z]$ then\\
\> \>  \> \>  \> \{ $H$.insertOrAdjust $(z, D[z])$;\\
\> \>  \> \>  \> $G$.insertOrAdjust($z, D[z]+outWeight[z])$;\}\\

  \>  \> {\bf endwhile};\\
  \\
\> \> {\bf procedure} processEdge3($z,k$);\\
\>\> {\bf var} $changed$: boolean initially false;\\
\>\>\> $minD, minU$: int initially $\infty$;\\
\> \> // step 1:  edge relaxation\\
 \> \> {\bf if} ($D[k] > D[z]+w[z,k]$) then \\
 \> \> \>$D[k] := D[z] + w[z,k]$;\\
  \> \> \>$changed$ := true;\\
\\
\> \> // step 2: Update $C[v]$ for all predecessors $v$ of $k$ \\
\>\> {\bf forall} $v: \neg fixed[v] \wedge ((v,k) \in E)$\\
\> \>\>   $C[v] := \max (C[v], H$.getMin());\\
\\
\> \> // step 3: Update $C$ via Eqn. \ref{eqn:crule}\\
\> \>   $C[k] := \max (C[k], \min \{ C[v] + w[v,k] ~|~  ((v,k)  \in E) \})$;\\

\\
\> \> // step 4: check if vertex $k$ is fixed \\
 \> \>  {\bf if} $(C[k] = D[k])$ then \\
 \> \> \> $fixed[k] := true$;\\
 \> \> \>  $R := R \cup \{ k \};$\\
 \> \> \>  $G$.remove($k$); $H$.remove($k$);\\
    \>\>  else if   ($changed \wedge (k \not \in Q)$) then $Q$.insert$(k)$;

 \end{tabbing}
\end{minipage}
} 
\end{center}
\caption{Algorithm $SP_3$: Using upper bounds as well as lower bounds \label{fig:alg-sp3}}
\end{figure}
}

Our third algorithm $SP_3$ is shown in Fig. \ref{fig:alg-sp3}.
We first remove from the heap $H$ all those non-fixed vertices $j$ such that $D[j] \leq threshold$.
All these vertices are marked as fixed.
Also, whenever any vertex is added or removed from the heap $H$, we also apply the same operation on the heap $G$.
In $SP_3$, it is more convenient to keep only the non-fixed vertices in $G$ and $H$.
All the vertices that are marked as fixed are removed from both $G$ and $H$.
Note that the deletion from the heap is only a virtual operation. It simply corresponds to marking that vertex as fixed.
Whenever a vertex is removed from any of the heaps in the {\tt removeMin} operation, and it is a fixed vertex, the algorithm 
simply discards that vertex and continues. Hence, vertices are physically removed only via {\tt removeMin} operation.
The {\tt getMin} operation removes any fixed vertex via {\tt removeMin}, so that {\tt getMin} applies only to the non-fixed vertices.
Whenever vertices are removed from $H$ via {\tt removeMin} operation, 
they are inserted in $R$ which 
explores them using {\rm processEdge3}.

Whenever we process an edge $(z,k)$, we update $D[k]$ as well as $C[k]$.
If $C[k]$ and $D[k]$ become equal then $v_k$ is marked as a fixed vertex; otherwise,
if $D[k]$ has changed then it is inserted in $Q$ for later processing.

To update $C[k]$, we first apply Lemma \ref{lem:crule2} to all the non-fixed predecessors of $k$, and then use Eqn. \ref{eqn:crule} to update
$C[k]$.
To apply Lemma \ref{lem:crule2}, we set $C[v]$ for any non-fixed predecessor vertex $v$ as the maximum of its previous value and
$H.getMin()$.
The method {\rm processEdge3} takes time $O(max(\log n, \Delta))$ where $\Delta$ is the maximum in-degree of any vertex.

We now show that $SP_3$ generalizes $SP_2$ (which, in turn, generalizes $SP_1$).

\begin{theorem}
Any vertex marked fixed by $SP_2$ in any iteration is also fixed by $SP_3$ in that iteration or earlier.
\end{theorem}
\begin{proof}
$SP_2$ fixes a vertex when $pred[k]$ equals zero, or when $D[k] \leq d + inWeight[k]$.
When $pred[k]$ equals zero, all the predecessors of $v_k$ are fixed and their $C$ value matches their $D$ value. Therefore, 
$C[k] := \max (C[k], \min \{ C[v] + w[v,k] ~|~  ((v,k)  \in E) \})$ guarantees that
$C[k] \geq  \min \{ D[v] + w[v,k] ~|~  ((v,k)  \in E)  = D[k]$. Therefore, vertex $k$ is marked as fixed.

Now suppose that $D[k] \leq d + inWeight[k]$ in $SP_2$.
Let $z$ be the vertex that discovered $k$ in $SP_2$. Then $D[k] \leq d + inWeight[k]$ implies
$D[k] \leq \min(D[k], d + inWeight[k])$.
Since $D[k] \leq D[z] + w[z,k]$, we get that
 $D[k] \leq \min((D(z)+ w[z,k]), d + inWeight[k]$.
From the definition of $inWeight[k]$, we get that
$D[k] \leq \min((D(z)+ w[z,k]), d + \min \{ w[v,k] ~|~ (v,k) \in E \wedge (v \neq z) \} )$.
Since $z$ is a fixed vertex, we get
$D[k] \leq \min((C(z)+ w[z,k]),  \min \{ d + w[v,k] ~|~ (v,k) \in E \wedge (v \neq z) \} )$.
Since $C[v]$ for all predecessors of $k$ is set to at least $d$ in step 2 of $SP_3$, we get that
$D[k] \leq \min((C(z)+ w[z,k]),  \min \{ C[v] + w[v,k] ~|~ (v,k) \in E \wedge (v \neq z) \} )$.
By combining two arguments of the $\min$, we get 
$D[k] \leq \min \{ C[v] + w[v,k] ~|~ (v,k) \in E \} $.
 The right hand side is $C[k]$ due to assignment of $C[k]$ at step 4.
Since $D[k] \leq C[k]$, we get that vertex $k$ is marked as fixed.
\end{proof}

We now show that any vertex marked fixed by out-version or in-version of \cite{crauser1998parallelization} is also marked fixed by $SP_3$.
\begin{lemma}
(a)  $SP_3$ fixes any vertex $k$ such that $D[k] \leq  \min \{ D[x] + w[x,y] ~|~ \neg fixed(x) \}$.\\
(b) $SP_3$ fixes any vertex $k$ such that  $D[k] \leq \min \{D[y] ~|~ \neg fixed(y) \} + \min \{ w[v,k] ~|~ (v,k) \in E \}$.
\end{lemma}
\begin{proof}
(a) follows from $threshold$ computed and marking of vertices as fixed based on that.\\
(b) Suppose $D[k] \leq \min \{D[y] ~|~ \neg fixed(y) \} + \min \{ w[v,k] ~|~ (v,k) \in E \}$.
The first part of the sum is equal to $H.getMin()$ due to the property of $H$.
Therefore, this expression is equal to $ \min \{ H.getMin() + w[v,k] ~|~ (v,k) \in E \}$.
From Step 2 in $SP_3$, this expression is at most
$ \min \{ C(v) + w[v,k] ~|~ (v,k) \in E \}$. From step 3, we get this expression to be at most
$ C[k] $. Therefore, $D[k] \leq C[k]$ and $k$ is fixed.

\end{proof}

The following Theorem summarizes properties of $SP_3$.
\begin{theorem}
Algorithm $SP_3$ computes the cost of the shortest path from the source vertex $v_0$ to all other vertices in  $O(n + e(\max(\log n, \Delta)))$ time, where
$\Delta$ is the maximum in-degree of any vertex.
\end{theorem}

\section{Algorithm $SP_4$: A Parallel Label-Correcting Algorithm}
In this section we present an algorithm when a large number of cores are available.
The goal of the algorithm is to decrease the value of $D$ and increase the value of $C$ in 
as few iterations of the {\em while} loop as possible. 
All our earlier algorithms explore only
fixed vertices with the motivation of avoiding multiple edge-relaxation of the same edge (in the spirit of Dijkstra's algorithm).
In contrast, $SP_4$ is a label-correcting algorithm that relaxes as many edges as possible in each iteration (in the spirit of Bellman-Ford algorithm).
Similarly, it recomputes $C$ for as many vertices as possible and terminates faster than $SP_3$.

{\small
\begin{figure}[htbp]\begin{center}
\fbox{\begin{minipage}  {\textwidth}\sf
\begin{tabbing}
\=x\=xxx\=xxx\=xxx\=xxx\=xxx\=xxx\=\kill
\>\> {\bf var} $D$: array[$0 \ldots n-1$] of integer\\
\>\>\> initially $\forall i: D[i] := \infty$;\\
\>\> \> $fixed$: array[$0 \ldots n-1$] of boolean\\
\>\>\>\> initially $\forall i: fixed[i] := false$;\\
\>\> \> $C$: array[$0 \ldots n-1$] of integer\\
\>\>\>\> initially $\forall i: C[i] = 0$;\\
\> \> \> $outWeight$: array[$0 \ldots n-1$] of integer initially\\
\>\>\>\> $\forall i: outWeight[i] = \min \{ w[i,j] ~|~ (i,j) \in E\}$;\\
\>\> \> $Dout$: array[$0 \ldots n-1$] of integer\\
\>\>\>\> initially $\forall i: Dout[i] = \infty$;\\
\> \> \>  int $threshold$; \\
\> \> \>  int $minD$; \\
\\
\> \>  $D[0] := 0;$\\
\> \>  $Dout[0] := D[0] + outWeight[0];$\\
\>\>   boolean $changed := true;$\\
\\
\> \> {\bf while} ($changed \wedge (\exists i: \neg fixed[i] \wedge (D[i] < \infty)$) \\
\>\>\>    $changed := false;$\\
\\
\> \>  \> // Step 1: find the minimum value of $D$ and $Dout[x]$  \\
\> \> \>  $threshold$ := min $Dout[x]$ of all non-fixed vertices; \\
\> \> \>  $minD$ := min $D[x]$ of all non-fixed vertices; \\

\\
\> \>  \> // Step 2: Fix all vertices with $D[x] \leq threshold$\\
\> \> \>  {\bf forall} $x$ such that $(D[x] \leq threshold)$ {\bf in parallel}\\
\>\> \> \>   $fixed[j] := true;$\\
\>\> \> \>  $C[j] := D[j];$\\
\\
\> \>  \> // Step 3: Update $D$ values\\
\> \> \>  {\bf forall} $x,y$ such that $(D[x] < \infty)$\\
\>\>\>\>  $\wedge \neg fixed[y] \wedge ((x, y)  \in E)$ {\bf in parallel}\\
\>\> \> \> if $(D[y] > D[x] + w[x,y])$ then\\
\>\> \> \> \> $D[y] := D[x] + w[x,y]$;\\
\>\> \> \> \> $Dout[y] := D[y] + outWeight[y]$;\\
\>\> \> \> \> $changed := true$;\\
\\
\> \>  \> // Step 4: Update $C$ values\\
\> \> \>  {\bf forall} $y$ such that $\neg fixed[y]$ {\bf in parallel}\\
\> \> \> \>  $C[y] := \max (C[y], minD )$;\\
\> \> \>  {\bf forall} $y$ such that $\neg fixed[y]$ {\bf in parallel}\\
\> \> \> \>  $C[y] := \max (C[y],$ \\
\>\>\>\>\>   $\min \{ C[x] + w[x,y], (x, y)  \in E)\})$;\\
\\
\> \>  \> // Step 5: Update $fixed$ values\\
\> \> \>  {\bf forall} $y: \neg fixed[y] \wedge (D[y] < \infty)$ {\bf in parallel}\\
\>\> \> \>  if ($C[y] = D[y]$) $fixed[y] := true$;\\
\> \> {\bf endwhile};
 \end{tabbing}
\end{minipage}
} 
\end{center}
\caption{Algorithm $SP_4$: A Bellman-Ford Style Algorithm with both upper and lower bounds  \label{fig:alg-sp4}}
\end{figure}
}

The algorithm is shown in Fig. \ref{fig:alg-sp4}.
We use an outer while loop that is executed so long as $changed \wedge (\exists i: \neg fixed[i] \wedge (D[i] < \infty)$.
The variable $changed$ is used to record if any vertex changed its $D$ value. This is a well-known optimization of
Ford-Bellman algorithm for early termination.
If $D$ did not change in the last iteration of the while loop, we have reached the fixed point for $D$ and it is equal to 
$cost$. Even if $D$ changed for some vertices but all vertices are fixed, then their $D$ values cannot change and we can terminate the
algorithm. The conjunct $(D[i] < \infty)$ allows us to restrict the algorithm to examine only discovered vertices.

In step 1, we find $threshold$ equal to the minimum of all $Dout$ values of non-fixed vertices just as in $SP_3$.
We also find $minD$ equal to the minimum of all $D$ values for non-fixed vertices.
With $n$ processors this step can be done in $O(\log \log n)$ time and $O(n)$ work on a common-CRCW PRAM with the
standard technique of using a doubly logarithmic tree and cascading\cite{jaja1992introduction}.
In step 2, we fix all the vertices that have $D$ values less than or equal to the threshold. This step can be done in $O(1)$ time and $O(n)$ work.
In step 3, we first explore all the discovered vertices. All vertices adjacent to these vertices become discovered if 
they have not been discovered earlier. In addition, we also relax all the {\em incoming} edges to vertices that are not fixed.
Clearly, this is equivalent to relaxing
all edges as in the Bellman-Ford algorithm because for fixed vertices their $D$ value cannot decrease.
This step can be performed in $O(1)$ time and $O(e)$ work with $e$ cores.
In step 4, we compute lower bounds for all non-fixed vertices. 
We first update $C$ for all non-fixed vertices to be at least as large as $minD$.
In the second parallel step, we simply apply Eqn. \ref{eqn:crule} to update all $C$'s for
all non-fixed vertices. 
This step can also be performed in $O(1)$ time and $O(e)$ work with $e$ cores.
In step 5, we recompute the array $fixed$ based on $C$ and $D$. This step can be performed in $O(1)$ time
and $O(n)$ work.
The total number of iterations is at most $n$ giving us the parallel time complexity of $O(n \log \log n)$ and work 
complexity of $O(n e)$. The number of iterations in $SP_4$ algorithm is 
less than or equal to the number of iteration required by $SP_3$.
We now show the following property of $SP_4$.
\begin{theorem}
Algorithm $SP_4$ computes the cost of the shortest path from the source vertex $v_0$ to all other vertices in time $O(n \log \log n)$
and work $O(n e)$ with $e$ processors. 
\end{theorem}
\begin{proof}
We first show the correctness of $SP_4$.
It is sufficient to show that the while loop maintains the invariant that $D[x]$ is an upper bound and $C[x]$ is a lower bound on the cost
to the vertex $x$. Steps 1 and 2 correctly maintain $D$ follows from \cite{crauser1998parallelization}. 
Step 3 is the standard Bellman-Ford rule and it correctly maintains $D$. Step 4 correctly maintains $C$ due to Lemma \ref{lem:crule}.
Step 4, simply maintains the invariant that $fixed(x) \equiv (C[x] = D[x])$.

The time and work complexity follows from the earlier discussion.
\end{proof}


\section{Conclusions and Future Work}
In this paper, we have presented four algorithms for the shortest path problem. 
We present algorithms $SP_1$ and $SP_2$ that reduce the number of heap operations required by Dijkstra's algorithm and
allow exploration of multiple vertices in parallel thereby reducing its sequential bottleneck.
We also present algorithms $SP_3$ and $SP_4$ that require more work than Dijkstra's algorithm but reduce the sequential bottleneck even further.
These algorithms are the first ones that exploit 
both upper and lower bounds on the cost of the shortest path to increase the number of vertices that can be explored in parallel.
Extending these algorithms for distributed shared memory is a future research direction.

\bibliography{refs,refs2,refs3,sssp}

\end{document}